\documentclass[12pt, a4paper]{article}

\usepackage[margin=1in]{geometry} %page size
\usepackage{amsfonts, amscd, amssymb, mathtools, mathrsfs, dsfont, bbm, bbding} %math and other symbols
\usepackage[amsmath, amsthm, thmmarks]{ntheorem} %theorem style, auto qedhere
\usepackage{graphicx, xypic, color, float} %graph
\usepackage{indentfirst}%noindent
\usepackage{lmodern}%better output
\usepackage[T1]{fontenc} %better output
\usepackage{enumerate, listings, verbatim, paralist}%list%enumitem
\usepackage{setspace, xspace}%set space, shortward command
\usepackage[colorlinks=true, citecolor=blue]{hyperref}
\usepackage{extarrows}%used to define def
\usepackage{pdfpages}%insert pdf files
\usepackage{ulem}

\usepackage{setspace}
\AtBeginDocument{%
\addtolength\abovedisplayskip{-0.15\baselineskip}%
\addtolength\belowdisplayskip{-0.15\baselineskip}%
\addtolength\abovedisplayshortskip{-0.15\baselineskip}%
\addtolength\belowdisplayshortskip{-0.15\baselineskip}%
}

\newtheorem{thm}{Theorem} [section]
\newtheorem{prop}[thm]{Proposition}

\newtheorem{lemma}[thm]{Lemma}

\numberwithin{equation}{section}

\renewcommand{\geq}{\geqslant}
\renewcommand{\ge}{\geqslant}
\renewcommand{\leq}{\leqslant}
\renewcommand{\le}{\leqslant}
\newcommand{\citeprop}[1]{Proposition~\ref{#1}}

\newcommand{\citelem}[1]{Lemma~\ref{#1}}

\newcommand{\opfont}{\mathbb}

\newcommand{\BE}[2][]{\ensuremath{\operatorname{\opfont{E}}^{#1}\!\left[#2\right]}}
\newcommand{\bp}{\ensuremath{\opfont{P}}}
\newcommand{\bq}{\ensuremath{\opfont{Q}}}
\newcommand{\BP}[2][]{\ensuremath{\operatorname{\opfont{P}}^{#1}\!\left(#2\right)}}
\newcommand{\BF}{\ensuremath{\mathcal{F}}}

\newcommand{\R}{\ensuremath{\operatorname{\mathbb{R}}}}

\newcommand{\dd}{\ensuremath{\operatorname{d}\! }}
\newcommand{\dt}{\ensuremath{\operatorname{d}\! t}}
\newcommand{\ds}{\ensuremath{\operatorname{d}\! s}}

\newcommand{\ddp}{\ensuremath{\operatorname{d}\! p}}

\newcommand{\setq}{\mathscr{Q}^{+}}

\newcommand{\idd}[1]{\ensuremath{\operatorname{\mathds{1}}_{#1}}}

\newcommand{\essinf}{\ensuremath{\mathrm{ess\:inf\:}}}
\newcommand{\esssup}{\ensuremath{\mathrm{ess\:sup\:}}}
\newcommand{\barq}{\overline{Q}}
\newcommand{\barx}{\overline{X}}
\newcommand{\barp}{\bar{p}_{\lambda}}
\newcommand{\barh}{\overline{H}}
\newcommand{\ep}{\varepsilon}
\newcommand{\opl}{\Phi} 
 
\newcommand{\dsp}{\displaystyle}

\newcommand{\seta}{\mathscr{A}}

\newcommand{\pr}{\vartheta}
\newcommand{\qpr}{Q_{\pr}}

\def\tr{\intercal}
\newcommand{\al}{\alpha}
\newcommand{\lam}{\lambda}

\newcommand{\calJ}{\mathcal{J}}
\newcommand{\calL}{\mathcal{L}}

\newcommand{\vd}{\mathfrak{S}}

\begin{document}
%%%%%%%%%%%%%%%%%%%%%%%%%%%%%%%%%%%%%%%%%%%%%%%%
\title{$\al$-robust utility maximization with intractable claims: A quantile optimization approach
}
\author{Xinyu Chen\thanks{Department of Applied Mathematics, The Hong Kong Polytechnic University, Kowloon, Hong Kong SAR, China. Email: \url{24052241R@connect.polyu.hk}, \url{xyheartie@gmail.com}}
\and Zuo Quan Xu\thanks{Department of Applied Mathematics, The Hong Kong Polytechnic University, Kowloon, Hong Kong SAR, China. Email: \url{maxu@polyu.edu.hk}}}
\date{\today}
\maketitle 
\begin{abstract}
This paper studies an $\alpha$-robust utility maximization problem where an investor faces an intractable claim --- an exogenous contingent claim with known marginal distribution but unspecified dependence structure with financial market returns. The $\alpha$-robust criterion interpolates between worst-case ($\alpha=0$) and best-case ($\alpha=1$) evaluations, generalizing both extremes through a continuous ambiguity attitude parameter. For weighted exponential utilities, we establish via rearrangement inequalities and comonotonicity theory that the $\alpha$-robust risk measure is law-invariant, depending only on marginal distributions. This transforms the dynamic stochastic control problem into a concave static quantile optimization over a convex domain. We derive optimality conditions via calculus of variations and characterize the optimal quantile as the solution to a two-dimensional first-order ordinary differential equation system, which is a system of variational inequalities with mixed boundary conditions, enabling numerical solution. Our framework naturally accommodates additional risk constraints such as Value-at-Risk and Expected Shortfall. Numerical experiments reveal how ambiguity attitude, market conditions, and claim characteristics interact to shape optimal payoffs.
\end{abstract}

\section{Introduction}\label{sec:intro}

Portfolio selection constitutes a cornerstone of mathematical finance and economic theory. Since its modern formulation by Merton \cite{M69}, expected utility maximization (EUM) has provided the predominant normative framework for investment decisions under uncertainty. In a representative setting, an investor with initial endowment $x>0$ seeks an admissible trading strategy $\pi$ to maximize the expected utility of terminal wealth:
\begin{align}\label{eut}
\sup\limits_{\pi} \BE{u(X^{\pi}(T))},
\end{align}
where $X^{\pi}(\cdot)$ denotes the wealth process under strategy $\pi$, $u:\R_+ \to \R$ is an increasing and concave utility function representing investor's risk preference, and the expectation is taken under a reference probability measure $\bp$ that models the investor's beliefs about market dynamics. When market information is complete and the probabilistic structure is known, the problem \eqref{eut} can be solved using standard tools from stochastic optimal control \cite{KS98,YZ99}.

The classical EUM framework, however, relies on several important assumptions that limit its practical applicability. A principal limitation is its inability to address \emph{model ambiguity} --- situations where the investor lacks confidence in the exact probabilistic specification or faces uncertainty regarding parameter values and dependence structures. In response to this limitation, several generalizations have emerged:

\begin{description}
\item[Partial Information Models.] In practice, investors may not observe all relevant state variables and must infer information from observable processes. This leads to partial information control problems, typically addressed using filtering theory and adaptive control methods (see, e.g., \cite{XZ07,WWX13,XXZ21}).

\item[Ambiguity-Averse Robust Optimization.] To hedge against model misspecification, an investor may consider a family $\mathcal{Q}$ of plausible probability measures centered around a reference measure $\bp$. The robust counterpart of \eqref{eut} takes either a constrained form:
\[
\sup_{\pi \in \seta} \inf_{\bq \in \mathcal{Q}} \BE[\bq]{u(X^{\pi}(T))}, \quad \text{with } \mathcal{Q} = \{\bq: \mathbf{D}(\bq\|\bp) \leq \delta\},
\]
or a penalized form: 
\[
\sup_{\pi \in \seta} \inf_{\bq} \left\{\BE[\bq]{u(X^{\pi}(T))} + \beta \mathbf{D}(\bq\|\bp)\right\},
\]
where $\mathbf{D}$ is a divergence measure (e.g., relative entropy, $f$-divergence). Such formulations provide protection against worst-case scenarios within an ambiguity set \cite{GS1989,HS2008,FS2009}.

\item[$\al$-Maxmin Expected Utility.] Pure maxmin models are often criticized for excessive pessimism. Experimental and empirical evidence suggests that individuals exhibit heterogeneous and context-dependent attitudes toward ambiguity \cite{HT1991,D2016}. The $\al$-maxmin model \cite{M02,LLX19,YLZL25} interpolates between worst-case and best-case evaluations:
\[
(1-\al)\inf_{\bq \in \mathcal{Q}} \BE[\bq]{u(X^{\pi}(T))} + \al\sup_{\bq \in \mathcal{Q}} \BE[\bq]{u(X^{\pi}(T))}, \quad \al \in [0,1].
\]
The parameter $\al$ quantifies the investor's ambiguity attitude: $\al=1$ corresponds to extreme ambiguity seeking (maxmax), $\al=0$ to extreme ambiguity aversion (maxmin), and intermediate values represent mixed attitudes.
\end{description}

A common feature of the aforementioned models is their exclusive focus on \emph{endogenous} wealth generated from financial market investments. In practice, investors often hold \emph{exogenous} contingent claims --- such as lottery winnings, insurance liabilities, inheritance, or non-traded options --- whose payoff is not replicated by financial market assets. 
When the dependence between such claims and market returns is unknown or unspecified, they are termed ``intractable claims'' in the terminology introduced by Hou and Xu \cite{HX16} and later studied by Wang, Yue, and Huang \cite{WYH17}, Li, Xu and Zhou \cite{LXZ23}, and Maity, Bera, and Selvaraju \cite{MBS25}.
 For such claims, only the marginal distribution is available to the investor, while their joint dependence with market assets remains completely unspecified.

Li, Xu and Zhou \cite{LXZ23} study a worst-case robust EUM problem with an intractable claim $\pr$:
\begin{align}\label{defJ}
\inf\limits_{Y\sim \pr}\;\BE{u(X^{\pi}(T)+Y)},
\end{align}
where the infimum is taken over all random variables $Y$ with the same (marginal) distribution as $\pr$. While this formulation provides a conservative hedge against worst-case dependence, it may be overly pessimistic for investors with non-extreme ambiguity attitudes.

Related but distinct studies address utility maximization in incomplete markets with random endowments, contingent claims, or non-tradable assets \cite{D97,CSW01,HK04,HKS05}. In these models, the contingent claim is a known random variable or process; consequently, even when it is not hedgeable, the investor can acquire more information about it as time approaches maturity. In contrast, our work assumes knowledge only of the marginal distribution of the intractable claim, while its dependence structure with market returns remains completely unspecified. Thus, the intractable claim framework introduced by \cite{HX16} and further developed by \cite{LXZ23} serves as the more relevant benchmark for our analysis. Incorporating this realistic feature into portfolio optimization models is a problem of substantial practical and theoretical importance.

This paper bridges the gap between $\al$-maxmin preferences and intractable claims by studying the following \emph{$\al$-robust utility maximization} problem with intractable claim:
\begin{align*}%\label{eut+0}
\sup\limits_{\pi} J_{\al}(X^{\pi}(T)),
\end{align*}
where the \emph{$\al$-robust risk measure} $J_{\al}$ is defined as
\begin{align*}%\label{Jdef0}
J_{\al}(X):=(1-\al) \inf\limits_{Y\sim \pr}\;\BE{u(X+Y)}+\al\sup\limits_{Y\sim \pr}\;\BE{u(X+Y)}, \quad \al \in [0,1].
\end{align*}
This measure generalizes both extremes: it recovers the worst-case evaluation \eqref{defJ} when $\al = 0$ and the best-case evaluation when $\al = 1$. Thus, the parameter $\al \in [0, 1]$ introduces a continuous spectrum of ambiguity attitudes, interpolating between pure pessimism ($\al = 0$) and pure optimism ($\al = 1$).

The presence of the intractable claim $\pr$, whose joint dependence structure with the investment output $X^{\pi}(T)$ is unspecified, precludes the direct application of classical stochastic control methods. Our analysis addresses this challenge through the following methodological steps:

\begin{enumerate}
\item \textbf{Representation via Rearrangement Theory.} For a broad class of weighted exponential utilities (which extends the classical exponential utility), we employ rearrangement inequalities and comonotonicity theory \cite{DDGKV2002a,DDGKV2002b,M97} to obtain explicit representations of the extreme values in $J_{\al}$. These representations depend solely on the \emph{marginal} distributions of $X$ and $\pr$, transforming $J_{\al}$ into a law-invariant risk measure. This crucial simplification circumvents the need for any conditions on the joint dependence structure between the portfolio output and the intractable claim.

\item \textbf{Quantile Reformulation.} Exploiting the law-invariance of $J_{\al}$, we reformulate the $\al$-robust utility maximization problem as a static \emph{quantile optimization problem}. Specifically, we optimize over the set of admissible quantiles of terminal wealth. This transformation possesses two key advantages: (i) it converts the original non-concave dynamic problem into a concave maximization over a convex domain, and (ii) it translates the probabilistic constraints into tractable analytical conditions on quantiles.

\item \textbf{Optimality Characterization.} We derive first-order necessary and sufficient conditions for optimality in the quantile formulation. The optimal quantile is characterized as the solution to a novel two-dimensional first-order ordinary differential equation (ODE) system. While closed-form solutions are generally unavailable due to the equation's nonlinear structure, the ODE system provides a rigorous foundation for numerical analysis and reveals the structural properties of optimal policies.

\item \textbf{Numerical Analysis and Financial Implications.} We develop a numerical scheme to solve the characterizing ODE system and perform comprehensive sensitivity analyses. Our numerical studies yield several new financial insights regarding the interplay between ambiguity attitude, risk aversion, and the presence of intractable claims.
\end{enumerate}

The remainder of this paper is organized as follows. This section concludes with some preliminary notations and conventions. In Section~\ref{sec:pf}, we present the financial market model, define the class of admissible trading strategies, and review key results from rearrangement theory that are essential for our analysis. Section~\ref{sec:qf} analyzes the $\al$-robust risk measure $J_{\al}$, establishing its law-invariance and key properties. Section~\ref{sec:solution} develops the quantile formulation of the main problem and derives the associated optimality conditions. Section~\ref{sec:numerical} performs numerical studies based on the theoretical results obtained in the previous section. Finally, Section~\ref{sec:conclusion} concludes and suggests directions for future research. 

\paragraph{Notations and Conventions.}
For any cumulative distribution function $F$, we define its associated quantile (or generalized inverse) function as
\[
Q_{F}(p) := \inf\big\{z\in\R \mid F(z) > p\big\}, \quad p \in (0,1),
\]
which is the right-continuous inverse of $F$. Since $Q_F$ is non-decreasing, we may extend its domain by setting
\[
Q_{F}(0) := \lim_{p\downarrow 0} Q_{F}(p) \quad \text{and} \quad Q_{F}(1) := \lim_{p\uparrow 1} Q_{F}(p).
\]
For a random variable $Y$, denote by $F_{Y}$ its distribution function and by $Q_{Y} := Q_{F_Y}$ its quantile. One readily verifies that
$Q_Y(0) = \essinf Y$ and $Q_Y(1) = \esssup Y$. We write $X \sim Y$ to indicate that two random variables $X$ and $Y$ have the same distribution, and $X \sim F$ if $X$ follows distribution $F$.

Every quantile is right-continuous and non-decreasing (henceforth abbreviated as \emph{RCI}). Consequently, it admits a left limit at every point. Conversely, any RCI function $Q: (0,1) \to \R$ is the quantile of the random variable $Q(U)$, where $U \sim \mathrm{Uniform}(0,1)$. Thus, the class of quantiles coincides with the class of RCI functions on $(0,1)$.

In the sequel, the qualifiers ``almost everywhere'' (a.e.) and ``almost surely'' (a.s.) may be omitted when no confusion can arise. Finally, for a matrix or vector $M$, we denote by $M^{\tr}$ its transpose, and by $|M|$ the Frobenius norm $\sqrt{\operatorname{trace}(MM^{\tr})}$.

\section{Problem Formulation}\label{sec:pf}
Throughout this paper, we fix a probability space $(\Omega, \mathbb{F}, \bp)$ satisfying the usual assumptions, along with a standard $m$-dimensional Brownian motion $$W=\{(W_1(t), \cdots, W_m(t))^{\tr}, \;t\geq0\}$$ representing the uncertainties/risks of the financial market under consideration (to be described below). Unless otherwise stated, 
a random variable is the shorthand for an $\mathbb{F}$-measurable random variable. We also fix an investment horizon $[0, T]$ where $T>0$ is a constant maturity. Let $\{\BF_t\}_{t\geq0}$ be the filtration generated by $W$ complemented by all the $\bp$-null sets, 
and $L^{0}_{\BF_{T}}$ be the set of $\BF_{T}$-measurable random variables.
We stress that $\BF_{T}\subsetneqq\mathbb{F}$; so there is randomness \emph{outside} of the financial market or, equivalently, not every random variable belongs to $L^{0}_{\BF_{T}}$.

\subsection{Financial Market Model}
\noindent
Consider a continuous-time arbitrage-free financial market consisting of $m+1$ assets traded continuously on the time horizon $[0,T]$. One asset is a risk-free bond whose price $S_0(\cdot)$ evolves according to the ordinary differential equation
\begin{equation*}
\begin{cases}
\dd S_0(t) = r(t) S_0(t) \dt, & t \in [0,T], \\[4pt]
S_0(0) = s_0 > 0,
\end{cases}
\end{equation*}
where $r(t)$ denotes the instantaneous risk-free rate at time $t$. The remaining $m$ assets are risky stocks whose prices satisfy the stochastic differential equations
\begin{equation*}
\begin{cases}
\dd S_i(t) = S_i(t)\Big[b_i(t)\dt + \displaystyle\sum_{j=1}^m \sigma_{ij}(t) \dd W_j(t)\Big], & t \in [0,T], \\[4pt]
S_i(0) = s_i > 0,
\end{cases}
\end{equation*}
for $i = 1,\dots,m$. Here $b_i(t)$ is the appreciation rate of stock $i$, $\sigma_{ij}(t)$ are volatility coefficients. 

Define the appreciation rate vector $b(t) := (b_1(t),\dots,b_m(t))^\tr$, the volatility matrix $\sigma(t) := (\sigma_{ij}(t))_{m\times m}$, and the excess return vector $b(t) - r(t)\mathbf{1}_m$, where $\mathbf{1}_m$ denotes the $m$-dimensional vector of ones. Assuming $\sigma(t)$ is invertible for all $t$, we introduce the market price of risk process $\theta := \{\theta(t), t\in[0,T]\}$ by
\begin{equation*}
\theta(t) := \sigma(t)^{-1}\big(b(t) - r(t)\mathbf{1}_m\big), \quad t\in[0,T].
\end{equation*}
\textbf{Parameter Specification.}
Following \cite{KS98}, we impose the following standing assumptions throughout the paper:
\begin{compactitem}
\item The processes $r$, $b$ and $\sigma$ are $\{\mathcal{F}_t\}$-progressively measurable;
\item The interest rate $r$ is lower bounded, and the appreciation rate $b$ satisfies $\int_0^T |b(t)| \dt < \infty$ almost surely;
\item The market price of risk $\theta$ is not identically zero and satisfies $\int_0^T |\theta(t)|^2 \dt < \infty$ almost surely;
\item The stochastic exponential process
\[
Z(t) := \exp\left(-\int_0^t \theta(s)^\tr \dd W(s) - \frac{1}{2}\int_0^t |\theta(s)|^2 \ds\right), \quad t\in[0,T],
\]
is a true martingale.
\end{compactitem}
The last condition is guaranteed, for instance, by the Novikov condition $$\mathbb{E}\big[\exp\big(\frac{1}{2}\int_0^T |\theta(s)|^2 \ds\big)\big] < \infty.$$

These assumptions ensure that the financial market is \emph{standard and complete} in the sense of \cite[Definition 5.1, page 17]{KS98} and \cite[Theorem 6.6, page 24]{KS98}. Consequently, every contingent claim admits a unique replicating portfolio, and the state price density is uniquely determined.

\subsection{Investment Problem}
\noindent
Consider a small investor (``she'') whose transactions do not influence asset prices. She has an initial endowment $x>0$ and invests in the financial market over $[0,T]$. Let $\pi_i(t)$ denote the total market value of her wealth invested in stock $i$ at time $t$, $i=1,\dots,m$. Short selling is allowed, so $\pi_i(t)$ may be negative. Trading occurs continuously in a self-financing manner (no consumption or income) within a frictionless market (no transaction costs).

Define the portfolio process $\pi := \{(\pi_1(t),\dots,\pi_m(t))^\tr, t\in[0,T]\}$. The corresponding wealth process $X^\pi$ evolves according to the stochastic differential equation (see \cite{KS98})
\begin{align}\label{eq:state-positive}
\begin{cases}
\dd X^\pi(t) = \big[r(t)X^\pi(t) + \pi(t)^\tr \sigma(t) \theta(t)\big]\dt + \pi(t)^\tr \sigma(t) \dd W(t), & t\in[0,T], \\[4pt]
X^\pi(0) = x > 0.
\end{cases}
\end{align}

The investor makes decisions based on the partial information filtration $\{\mathcal{F}_t\}_{t\in[0,T]}$, which is strictly coarser than the full filtration. Thus, our model falls within the partial information framework.\medskip\\
\textbf{Admissible Portfolios.} A portfolio $\pi$ is called \emph{admissible} if:
\begin{compactitem}
\item It is $\{\mathcal{F}_t\}_{t\in[0,T]}$-progressively measurable;
\item It satisfies $\int_0^T |\sigma(t)^\tr \pi(t)|^2 \dt < \infty$ almost surely;
\item The corresponding wealth process $X^\pi$ from \eqref{eq:state-positive} exists and remains nonnegative almost surely.
\end{compactitem}
By no-arbitrage, the last condition is equivalent to requiring $X^\pi(T) \ge 0$ almost surely. For any admissible portfolio, \eqref{eq:state-positive} admits a unique continuous solution $X^\pi$, and $(X^\pi,\pi)$ is called an \emph{admissible pair}. Due to linearity of \eqref{eq:state-positive}, the set of admissible pairs is convex. Unless otherwise specified, we consider only admissible portfolios henceforth. \medskip\\
\textbf{Intractable Claim and Objective.} In addition to investment gains, the investor receives at maturity $T$ an intractable claim $\pr$ (e.g., lottery winnings, insurance liabilities, inheritance, or non-traded options) whose payoff is not replicated by market assets. Because the joint dependence between $\pr$ and market returns is unknown, evaluating expected utility of total terminal wealth $\mathbb{E}[u(X^\pi(T)+\pr)]$ is not meaningful.

As motivated in the introduction, we study the $\alpha$-robust utility maximization problem with intractable claim:
\begin{align}\label{eut+}
\sup_{\pi} J_\alpha(X^\pi(T)),
\end{align}
where $J_\alpha$ is the $\alpha$-robust risk measure: 
\begin{align}\label{Jdef}
J_\alpha(X)= (1-\alpha)\inf_{Y\sim\pr}\mathbb{E}[u(X+Y)] + \alpha\sup_{Y\sim\pr}\mathbb{E}[u(X+Y)], \quad \alpha\in[0,1]. \smallskip
\end{align}
\textbf{Utility Specification.} To ensure tractability, we focus on weighted exponential utility functions of the form
\begin{align}\label{WEU}
u(x) = -\int_0^\infty e^{-\gamma x} \dd F(\gamma), \quad x\in\mathbb{R},
\end{align}
where $F$ is a non-zero distribution function supported on a compact subset of $(0,\infty)$. This class includes the classical exponential utility as a special case and admits the convenient analytical properties:
\begin{compactitem}
\item $u$ is $C^\infty$-smooth with $k$-th derivative: 
\begin{align}\label{uk1}
u^{(k)}(x) = (-1)^{k+1}\int_0^\infty \gamma^k e^{-\gamma x} \dd F(\gamma), \quad x\in\mathbb{R};
\end{align}
\item $\lim_{x\to-\infty} u'(x) = +\infty$ and $\lim_{x\to+\infty} u'(x) = 0$ (by monotone convergence);
\item $u^{(k)}$ is strictly decreasing for odd $k$ and strictly increasing for even $k$; consequently, $u$ is strictly increasing and strictly concave.\medskip
\end{compactitem}
\textbf{Intractable Claim Specification.} We assume the quantile function $Q_\pr$ of the intractable claim $\pr$ is Lipschitz continuous and differentiable on $[0,1]$, which implies boundedness of $\pr$. This regularity facilitates the subsequent analysis.
\medskip\\
\textbf{Methodological Approach.} The unknown dependence between $\pr$ and $X^\pi(T)$ precludes classical stochastic control methods (e.g., stochastic maximum principle, dynamic programming in \cite{YZ99}). Following \cite{HX16,LXZ23}, we first transform problem \eqref{eut+} into a static optimization problem via the \emph{quantile reformulation}, which we develop in the next section.

\section{Quantile Formulation}\label{sec:qf}

Problem \eqref{eut+} is a dynamic stochastic control problem that falls outside the scope of classical methods such as the stochastic maximum principle or dynamic programming (see \cite{YZ99}). Following \cite{HX16,LXZ23}, we first transform it into a static quantile optimization problem and then apply the quantile method. The quantile method is a powerful tool in studying behavioral financial models; we refer to \cite{HZ11,X16,JZ08,B2015,WX2021,CSZ2024} and references therein for recent developments. 

\subsection{A Static Optimization Problem}

Employing the well-established martingale approach, we decompose the solution of \eqref{eut+} into two steps. First, we solve the static optimization problem 
\begin{eqnarray}\label{static1}
\begin{array}{rl}
\sup\limits_{X\in L^{0}_{\BF_{T}}} \quad & J_{\al}(X) \\ [2mm]
\mathrm{subject\ to} \quad & X\in\seta_{x}, 
\end{array}
\end{eqnarray}
where $\seta_x$ denotes the set of all attainable nonnegative terminal wealth levels with initial endowment $x>0$. By linearity of the wealth SDE \eqref{eq:state-positive}, $\seta_x$ is convex. Second, we find an admissible portfolio $\pi^*$ that replicates the optimal terminal wealth $\barx$; this replication step is standard in a complete market (see \cite{KS98}). Hence, our primary focus is on problem \eqref{static1}.

To proceed, we require a tractable characterization of $\seta_x$. The following classical result provides such a description (see \cite{KS98}).

\begin{lemma}\label{martingalemethod}
For any $x>0$, the set $\seta_x$ satisfies
\begin{align}\label{seta}
\big\{X \in L^0_{\mathcal{F}_T} : \mathbb{E}[\rho X] = x,\; X \ge 0\big\} \subseteq \seta_x \subseteq \big\{X \in L^0_{\mathcal{F}_T} : \mathbb{E}[\rho X] \le x,\; X \ge 0\big\},
\end{align}
where $\rho$ is the pricing kernel (also called stochastic discount factor) defined by
\begin{equation*}
\rho := \exp\left(-\int_0^T \Big(r(s) + \frac{1}{2}|\theta(s)|^2\Big)\ds - \int_0^T \theta(s)^\tr \dd W(s)\right).
\end{equation*}
\end{lemma}

The inequality $\mathbb{E}[\rho X] \le x$ is known as the \emph{budget constraint}. Under our standing assumptions, $\theta$ is not identically zero, so $\rho$ is a non-degenerate random variable. Moreover, since the stochastic exponential in $\rho$ is a martingale (hence has expectation one at all times) and $r$ is lower bounded, we have $\mathbb{E}[\rho] < \infty$.
\medskip\\
\textbf{Pricing Kernel Specification.}
To avoid technical complications, we impose the following regularity conditions on $\rho$:
\begin{compactitem}
\item The quantile function $Q_\rho$ of $\rho$ is continuously differentiable;
\item $Q_\rho(0) = 0$ and $Q_\rho(1) = +\infty$ (so $\rho$ is unbounded above). \medskip
\end{compactitem}
These conditions are satisfied, for instance, in the Black-Scholes market where $\rho$ follows a log-normal distribution. Under these assumptions, the distribution function $F_\rho$ of $\rho$ is continuous and satisfies the probability integral transform: $F_\rho(\rho) \sim \mathrm{Uniform}(0,1)$ (see \cite{X14}).

\subsection{Quantile Reformulation}

We now employ the quantile formulation method, which replaces the decision variable $X$ (a random variable) with its quantile function (a deterministic function). This transformation exploits the law-invariance of $J_\alpha$ and yields a tractable optimization problem over quantile functions. The following famous rearrangement inequality is fundamental to this reformulation (see \cite{DDGKV2002a}). 
\begin{lemma}\label{maxminlemma}
Given a random variable $X$ and a probability distribution function $\mu$, we have 
\begin{align}\label{max}
\max_{Y\sim \mu}\mathbb{E}[XY] = \int_0^1 Q_X(p) Q_\mu(p) \ddp=\mathbb{E}[XY_1] 
\end{align}
and
\begin{align}\label{min}
\min_{Y\sim \mu}\mathbb{E}[XY] = \int_0^1 Q_X(p) Q_\mu(1-p) \ddp= \mathbb{E}[XY_2], 
\end{align}
provided the integrals are well-defined, 
where $Y_1$, $Y_2 \sim\mu $, $X$ and $Y_1$ are comonotonic, and $X$ and $Y_2$ are anti-comonotonic.
\end{lemma}

By this result, we can determine both $\displaystyle \inf_{Y\sim \pr}\mathbb{E}[u(X+Y)]$ and $\displaystyle\sup_{Y\sim \pr}\mathbb{E}[u(X+Y)]$, taking the advantage of that $u$ is a weighted exponential utility.

\begin{lemma}\label{Jproperties}
We have
\begin{align}\label{Jexp}
J_{\al}(X) = \int_0^1 V(Q_X(p), p) \ddp, 
\end{align}
where
\begin{align}\label{Vdef}
V(x, p) :=(1-\al) u\big(x+\qpr(p)\big)+\al u\big(x+\qpr(1-p)\big), 
~~x\in\R, ~p\in(0, 1).
\end{align}
Also, the functional $J_{\al}$ is law-invariant in the sense that $J_{\al}(X)=J_{\al}(Y)$ if $X\sim Y$, and strictly increasing in the sense that
$J_{\al}(X_1)> J_{\al}(X_2)$
if $X_1\geq X_2$, $\BP{X_1>X_2}>0$. In particular, \[J_{\al}(X+\ep)> J_{\al}(X)\]
for any constant $\ep>0$. 
\end{lemma}

\begin{proof} 
By \cite{X14}, there exists a $U \sim \mathrm{Uniform}(0,1)$ such that $X=Q_X(U)$. 
For any constant $\gamma>0$, the two random variables $e^{-\gamma X}$ and $e^{-\gamma Y}$ are comonotonic if and only if $ X$ and $Y$ are comonotonic. Notice $Q_X(U)$ and $\qpr(U)$ are comonotonic and $\qpr(U)\sim \pr$, so by 
\citelem{maxminlemma}, we have 
\begin{align*} 
\sup\limits_{Y\sim \pr}\;\BE{e^{-\gamma(X+Y)}}&=
\sup\limits_{Y\sim \pr}\;\BE{e^{-\gamma(Q_X(U)+Y)}}= \BE{e^{-\gamma\, (Q_X(U)+\qpr(U))}}.
\end{align*}
Hence, by Fubini's theorem, 
\begin{align*}
\inf\limits_{Y\sim \pr}\;\BE{u(X+Y)}
&=-\int_{0}^{\infty} \sup\limits_{Y\sim \pr}\;\BE{e^{-\gamma(X+Y)}}\dd F(\gamma) \\
&= -\int_{0}^{\infty}\BE{e^{-\gamma(Q_X(U)+\qpr(U))}}\dd F(\gamma)\\
&=\BE{ -\int_{0}^{\infty}e^{-\gamma(Q_X(U)+\qpr(U))}\dd F(\gamma)}\\
&=\BE{u\big(Q_X(U)+\qpr(U)\big)}\\
&=\int_0^1 u\big(Q_X(p)+ \qpr(p)\big) \ddp.
\end{align*} 
Similarly, we can prove 
\begin{align*}
\sup\limits_{Y\sim \pr}\BE{u(X+Y)}&=\int_0^1 u\big(Q_X(p)+ \qpr(1-p)\big)\ddp.
\end{align*} 
Recalling the definition of $J_\al(X)$ in \eqref{Jdef}, we conclude 
\[
J_\al(X) = \int_0^1 \left[ (1-\al) u\big(Q_X(p)+ \qpr(p)\big) + \al u\big(Q_X(p)+ \qpr(1-p)\big) \right] \ddp, 
\]
namely, \eqref{Jexp} holds. 
This evidently implies that $J_{\al}$ is a law-invariant functional.

Suppose $X_1 \geq X_2$ and $\mathbb{P}(X_1 > X_2) > 0$. 
Then by definition one can see 
$Q_{X_1}(p) \geq Q_{X_2}(p)$ for all $p \in (0, 1)$, and there exists an interval $(a, b) \subset (0, 1)$ such that $Q_{X_1}(p) > Q_{X_2}(p)$ for $p\in(a,b)$, thanks to the right-continuity of quantiles.
Because $u$ is strictly increasing and 
$\al\in[0,1]$, we have, for all $p\in(0,1)$, 
\begin{align*}
V(Q_{X_1}(p), p) &= \al u\big(Q_{X_1}(p)+\qpr(1-p)\big)+ (1-\al)u\big(Q_{X_1}(p)+\qpr(p)\big) \\
&\geq \al u\big(Q_{X_2}(p)+\qpr(1-p)\big)+ (1-\al)u\big(Q_{X_2}(p)+\qpr(p)\big) \\
&=V\big(Q_{X_2}(p), p\big).
\end{align*}
Notice the above inequality is strict for $p\in(a,b)$.
Therefore, recalling that $V< 0$, we have that 
\begin{align*}%\label{increase}
J_{\al}(X_1) = \int_{0}^{1} V\big(Q_{X_1}(p), p\big) \ddp > \int_{0}^{1} V\big(Q_{X_2}(p), p\big) \ddp = J_{\al}(X_2).
\end{align*} 
The proof is complete. 
\end{proof}

Define a functional on quantiles as follows:
\begin{align*}
\calJ_{\al}(Q) 
:=&~(1-\al) \int_0^1 u\big(Q(p)+\qpr(p)\big) \ddp+\al \int_0^1 u\big(Q(p)+\qpr(1-p)\big) \ddp\\
=&~\int_0^1V(Q(p), p) \ddp.
\end{align*}
The above lemma shows $J_\al(X)=\calJ_{\al}(Q_X)$.

\begin{lemma}[Monotonicty and Concavity of $\calJ_\al$]\label{lemma:J_concave_Q} 
The functional $\calJ_\al$ is strictly increasing and strictly concave.
\end{lemma}

\begin{proof} 
This is an immediate consequence of the fact that $V(x,p)$ defined in \eqref{Vdef} is strictly increasing and strictly concave with respect to $x$.
\end{proof}

Consider the following relaxed optimization problem, for $x>0$:
\begin{eqnarray}\label{static2}
\begin{array}{rl}
\sup\limits_{X\in L^{0}_{\BF_{T}} } &\quad J_{\al}(X), \\ [2mm]
\mathrm{subject\ to} &\quad \BE{\rho X}\leq x, \; X\geq 0.
\end{array}
\end{eqnarray}

By the strict monotonicity of $J_\alpha$, any optimal solution $\barx$ to \eqref{static2}, if it exists, must satisfy the budget constraint with equality: $\mathbb{E}[\rho \barx] = x$. Consequently, $\barx \in \seta_x$ by Lemma \ref{martingalemethod}, and thus $\barx$ also solves the original static problem \eqref{static1}. Moreover, any admissible portfolio replicating $\barx$ is optimal for the dynamic problem \eqref{eut+}. Hence, solving \eqref{static2} suffices for our purposes.

As evident from the representation \eqref{Jexp}, the problem \eqref{eut+} lies outside the scope of classical stochastic control theory (e.g., the stochastic maximum principle or dynamic programming). To overcome this difficulty, we employ a quantile formulation approach, building upon the methodology developed in \cite{X23,X25}.

Let $\setq$ denote the set of all quantiles generated by nonnegative random variables, that is, 
\begin{align*}
\setq :=\big\{Q: (0, 1)\to \R \;\big|\; &\text{$Q$ is the quantile for some}\\
&\text{nonnegative random variable $X\in L^{0}_{\BF_{T}}$}\big\}.
\end{align*}
It is easy to see
\begin{align*}
\setq&=\big\{Q: (0, 1)\to [0, \infty) \;\big|\; \text{$Q$ is RCI}\big\},
\end{align*}
which clearly implies that $\setq$ is a convex set.

The preceding arguments show that solving \eqref{static2} reduces to the following problem:
\begin{equation}\label{static3}
\begin{aligned}
\sup\limits_{X\in L^{0}_{\BF_{T}} } &\quad J_{\al}(X)=\int_0^1 V(Q_X(p), p) \ddp=\calJ_{\al}(Q_X). \\ 
\mathrm{subject\ to} &\quad \BE{\rho X}= x, \; X\geq 0.
\end{aligned}
\end{equation}

According to \cite{X14}, any optimal solution $\barx$ to \eqref{static3} and the pricing kernel $\rho$ must be anticomonotonic. This crucial observation allows us to reformulate \eqref{static3} as a quantile optimization problem:
\begin{equation}\label{static3a}
\begin{aligned}
\sup_{Q\in\setq} & \quad \calJ_\al(Q)=\int_0^1 V(Q(p), p)\ddp \\
\text{subject to} & \quad \int_0^1 Q(p) Q_\rho(1-p)\ddp = x.
\end{aligned}
\end{equation}
Specifically, a quantile $\barq\in \mathcal{Q}^+$ solves \eqref{static3a} if and only if $\barx$ solves \eqref{static3} and $\barq = Q_{\barx}$.

Since $\mathcal{J}_\alpha(Q) \leq 0$ for all admissible $Q$, the problem \eqref{static3a} is well-posed with a finite optimal value. Moreover, $\mathcal{J}_\alpha$ is a strictly concave functional and the feasible set of quantiles is convex; consequently, \eqref{static3a} is a concave optimization problem admitting at most one optimal solution.

To solve \eqref{static3a}, we employ the Lagrange multiplier method. For any Lagrange multiplier $\lambda > 0$, consider the unconstrained problem:
\begin{equation}\label{static4}
\sup_{Q\in \setq} \calL(Q),
\end{equation}
where 
the Lagrangian functional $\calL$ is defined as 
\begin{align*}
\calL(Q) :=&~\calJ_\al(Q)+\lam\left(x-\int_0^1 Q(p) Q_\rho(1-p)\ddp\right) \\
=&~ \int_0^1 \big[V(Q(p), p)-\lam Q(p) Q_\rho(1-p)\big] \ddp+\lam x.
\end{align*} 
Since \( \calJ_\al\) is a strictly concave functional, so is \( \calL \) and the problem \eqref{static4} admits at most one optimal solution. 

\subsection{Optimal Solutions to \eqref{static4} and \eqref{static2}}\label{sec:solution}

The concavity of \eqref{static4} and convexity of its feasible domain render the problem well-suited for variational methods. Using calculus of variations, we obtain the following necessary and sufficient optimality conditions, which uniquely characterize the optimal solution whenever it exists.
\begin{prop}\label{variation}
Suppose $\barq \in \setq$. Then $\barq$ is an optimal solution to \eqref{static4} if and only if it satisfies:
\begin{align}\label{varcond}
\int_{0}^{1} \left[ \frac{\partial V}{\partial x}(\barq(p), p)-\lam Q_\rho(1-p) \right] (Q(p)-\barq(p)) \ddp \leq 0, ~~ \forall\; Q \in \setq.
\end{align}
\end{prop}

\begin{proof}
($\Longrightarrow$) Suppose $\barq\in \setq$ is an optimal solution to \eqref{static4}. For any $Q \in \setq$ and $\ep \in (0, 1)$, define the perturbation quantile:
\[
Q_\ep(p) = \ep Q(p)+(1-\ep)\barq(p), \quad p \in (0, 1).
\]
By the convexity of $\setq$, we have $Q_\ep \in \setq$. The optimality of $\barq$ implies:
\[
\int_0^1 V(Q_\ep(p), p)\ddp-\lam \int_0^1 Q_\ep(p) Q_\rho(1-p)\ddp \leq \int_0^1 V(\barq(p), p)\ddp-\lam \int_0^1 \barq(p) Q_\rho(1-p)\ddp. 
\]
After rearrangement, dividing both sides by $\ep > 0$ and taking $\ep \to 0^+$, we obtain from Fatou's lemma that 
\begin{align*}
0 &\geq \liminf_{\ep \to 0^+} \frac{1}{\ep} \left[ \int_0^1 V(Q_\ep(p), p)-V(\barq(p), p) \ddp-\lam \int_0^1 (Q_\ep(p)-\barq(p)\big) Q_\rho(1-p)\ddp \right] \\
&\geq \int_{0}^{1} \liminf_{\ep \to 0^+} \frac{V(Q_\ep(p), p)-V(\barq(p), p)}{\ep} \ddp-\lam \int_{0}^{1} (Q(p)-\barq(p)) Q_\rho(1-p)\ddp\\
&=\int_{0}^{1} \left[ \frac{\partial V}{\partial x}(\barq(p), p)-\lam Q_\rho(1-p) \right] (Q(p)-\barq(p)) \ddp,
\end{align*} 
giving the desired inequality \eqref{varcond}. 

($\Longleftarrow$) Conversely, assume $\barq\in\setq$ satisfies \eqref{varcond} but is not optimal to \eqref{static4}. Then there exist a $ Q_1 \in \setq$ and a constant $c > 0$ such that 
\[
\calL(Q_1) > \calL(\barq)+c.
\]
For $\ep \in (0, 1)$, let $Q_\ep(p) = \ep Q_1(p)+(1-\ep)\barq(p)$. The concavity of the functional $\calL$ gives 
\[
\calL(Q_\ep) \geq \ep\calL(Q_1)+(1-\ep)\calL(\barq) \geq \calL(\barq)+c\ep.
\]
It hence follows 
\[
\liminf_{\ep \to 0^+} \frac{1}{\ep} \left[ \calL(Q_\ep)-\calL(\barq) \right] \geq c > 0.
\]
However, the dominated convergence theorem together with \eqref{varcond} yields:
\begin{align*}
\liminf_{\ep \to 0^+} \frac{1}{\ep} \left[ \calL(Q_\ep)-\calL(\barq) \right] &= \frac{\dd}{\dd\ep} \calL(Q_\ep) \Big|_{\ep=0}\\
&=\int_{0}^{1} \left[ \frac{\partial V}{\partial x}(\barq(p), p)-\lam Q_\rho(1-p) \right] (Q(p)-\barq(p)) \ddp\\
&\leq 0.
\end{align*}
This contradiction establishes optimality of $\barq$. The proof is complete. 
\end{proof}
\par
Although the condition \eqref{varcond} characterizes the optimal solution to the problem \eqref{static4}, 
it is hard to use it to find the optimal solution, because one would have to compare the candidate $\barq$ with all the other quantiles in $\setq$, a task as difficult as solving \eqref{static4}.

Before going further, we first show some properties of $V$. 
\begin{lemma}\label{onV}
For each $p\in(0,1)$, 
the function $x\to V(x,p)$ belongs to $C^{\infty}(\R)$, and satisfies $\frac{\partial V}{\partial x}>0$, $\frac{\partial^{2}V}{\partial x^{2}}<0$, and
\begin{gather} 
\lim_{x\to+\infty} \frac{\partial V}{\partial x}(x,p)=0,~~
\lim_{x\to-\infty} \frac{\partial V}{\partial x}(x,p)=+\infty. \label{ineq22}
\end{gather}
Moreover, for $(x,p)\in [0,\infty)\times (0,1)$, we have
\begin{gather}\label{onV1}
u(\qpr(0))\leq V(x,p)< 0, \\
0 < \frac{\partial V}{\partial x}(x,p) \leq u'(\qpr(0)),\label{ineq1}\\
u''(\qpr(0))\leq \frac{\partial^{2}V}{\partial x^{2}}(x,p)<0,\label{ineq2}
\end{gather}
so that $V$, $\frac{\partial V}{\partial x}$ and $\frac{\partial^{2}V}{\partial x^{2}}$ are all bounded on $[0,\infty)\times (0,1)$. 
Moreover, the inequalities in \eqref{onV1}-\eqref{ineq2} are strict when $x>0$. 
\end{lemma}

\begin{proof}
Since $u< 0$ and thanks to \eqref{uk1}, it holds 
\begin{align*}
V(x,p)&= (1-\alpha)\,u\bigl(x + \qpr(p)\bigr)+ \alpha\,u\bigl(x + \qpr(1-p)\bigr)<0,\\ 
(-1)^{k+1}\frac{\partial^{k} V}{\partial x^{k}}(x,p)&=(-1)^{k+1}\Big[ (1-\alpha)\,u^{(k)}\bigl(x + \qpr(p)\bigr)+ \alpha\,u^{(k)}\bigl(x + \qpr(1-p)\bigr)\Big]>0. 
\end{align*}
The second estimate clearly implies \eqref{ineq22} as $\lim_{x\to-\infty}u'(x)=+\infty$ and $\lim_{x\to+\infty}u'(x)=0$. The signs of $V$, $\frac{\partial V}{\partial x}$ and $\frac{\partial^{2}V}{\partial x^{2}}$ in \eqref{onV1}-\eqref{ineq2} follow. 

Now suppose $x\geq 0$. 
By the monotonicity of quantiles,
\begin{align}\label{ineq1111}
x + \qpr(p)\ge \qpr(p)\ge \qpr(0),\quad x + \qpr(1-p)\ge \qpr(1-p)\ge \qpr(0), 
\end{align}
so by the monotonicity of \(u\), \(u'\), and \(u''\), we get the remainder bounds in \eqref{onV1}-\eqref{ineq2}. 

When $x>0$, the lower bounds in \eqref{ineq1111} are strict, by the strict monotonicity of \(u\), \(u'\), and \(u''\), the estimates in \eqref{onV1}-\eqref{ineq2} become strict. 
The proof is complete. 
\end{proof}
\begin{lemma}\label{lemma:G_decreasing_via_IFT}
There exists a unique function $\vd: \R\times (0,1)\to \R$ such that 
\begin{align}\label{defG1}
\frac{\partial V}{\partial x}\bigl(\vd(x, p), p\bigr)
=\;x,~~x\in\R,~p\in(0,1).
\end{align}
For each fixed $p\in(0,1)$, the function \(x\mapsto \vd(x, p)\) is in $C^{\infty}(\R)$ and strictly decreasing on $\R$.
Also, the function $\frac{\partial \vd}{\partial p}$ is continuous on $\R\times (0,1)$.
\end{lemma}

\begin{proof}
Set
\[
\varphi(\xi,x, p)
:=\frac{\partial V}{\partial x}(\xi,p)-x.
\] 
Thanks to \citelem{onV}, there is a unique \( \vd(x,p)\) solving
\[
\varphi\bigl(\vd(x,p),\,x,\,p\bigr)=0.
\]
Since $\varphi$ is infinitely differentiable with respect to $\xi$ and $x$, 
by the implicit function theorem, for each fixed $p\in(0,1)$, the function \(x\mapsto \vd(x,p)\) belongs to $C^{\infty}(\R)$. Differentiating the above equation in $x$ yields 
\[
\frac{\partial \vd}{\partial x}
=-\frac{\frac{\partial\varphi}{\partial x}}{\frac{\partial\varphi}{\partial \xi}}
=\frac{1}{\frac{\partial^2V}{\partial x^2}\bigl(\vd(x,p),p\bigr)}<0,
\]
so the function \(x\mapsto \vd(x, p)\) is strictly decreasing on $\R$.

Notice $\frac{\partial\varphi}{\partial p}$ is continuously differentiable on $\R\times(0,1)$, so 
\[
\frac{\partial \vd}{\partial p}
=-\frac{\frac{\partial\varphi}{\partial p}}{\frac{\partial\varphi}{\partial \xi}}
\]
is continuous. This completes the proof.
\end{proof}

We next determine the value of optimal solution $\barq$ near zero which requires our assumption that $\rho$ is unbounded.

Let 
$$\barp=\sup\big\{p\in(0,1): \lam Q_\rho(1-p)\geq u'\left(\qpr(0)\right)\big\}.$$ 
Since $\lam>0$, $u'>0$, $Q_{\rho}(0)=0$ and $Q_{\rho}(1)=+\infty$, we have $0<\barp<1$. 
By the monotonicity and right-continuity of quantiles, we have $\lam Q_\rho(1-p)\geq u'\left(\qpr(0)\right)$ for $p\in(0,\barp]$ and $\lam Q_\rho(1-p)< u'\left(\qpr(0)\right)$ for $p\in(\barp,1)$.

\begin{lemma}\label{initial}
The optimal solution $\barq$ to the problem \eqref{static4}, if exists, must satisfy $\barq(p)=0$ for $p\in(0,\barp]$.
\end{lemma}
\begin{proof} 
Suppose the claim is wrong. Then,
since $\barq$ is nonnegative and monotone, there exists $\ep\in(0,\barp)$ such that 
$\barq(p)>0$ for $p\in[\ep,\barp]$.
We then get from \citelem{onV} that 
\[
\lam Q_\rho(1-p)-\frac{\partial V}{\partial x}(\barq(p), p)\geq \lam Q_\rho(1-p) -u'\left(\qpr(0)\right)\geq 0, ~~p\in (0,\ep],
\] 
and
\[
\lam Q_\rho(1-p)-\frac{\partial V}{\partial x}(\barq(p), p)> \lam Q_\rho(1-p) -u'\left(\qpr(0)\right)\geq 0, ~~p\in [\ep,\barp].
\]
Set $Q(p)=\barq(p)\idd{p\geq \barp}$. Then $Q \in\setq$ and it follows from \eqref{varcond} that 
\begin{align*}
0\geq &~\int_{0}^{1} \left[ \frac{\partial V}{\partial x}(\barq(p), p)-\lam Q_\rho(1-p) \right] (Q(p)-\barq(p)) \ddp \\
=&~\int_{0}^{\barp} \left[ \lam Q_\rho(1-p)-\frac{\partial V}{\partial x}(\barq(p), p) \right] \barq(p) \ddp\\
\geq &~\int_{\ep}^{\barp} \left[ \lam Q_\rho(1-p)-\frac{\partial V}{\partial x}(\barq(p), p) \right] \barq(p) \ddp.
\end{align*}
But this is impossible since the integrand in the last integral is positive. 
The proof is complete.
\end{proof}

Our next step is to find an equivalent condition to \eqref{varcond} that can be easily verified and utilized. To this end, let $\barq$ be an optimal solution to \eqref{static4} and set
\begin{align}\label{defH}
\barh(p):= -\int_{p}^{1} \left[\frac{\partial V}{\partial x}(\barq(s), s) \right] \ds, ~~p\in[0, 1],
\end{align}
and 
\begin{align}\label{defeta}
\eta(p):= - \int_{p}^{1} Q_\rho(1-s)\ds=- \int_{0}^{1-p} Q_\rho(s)\ds, ~~p\in[0, 1].
\end{align} 
Using $\BE{\rho}<\infty$ and \eqref{ineq1}, one can easily show that $\barh$ is Lipchitz continuous on $[0, 1]$. 
We take $\barh'$ as the right-continues version of the derivative function of $\barh$. Then since the jump points of $\barq$ is at most countable, we have 
\begin{align*} 
\label{eta}
\barh'(p)= \frac{\partial V}{\partial x}(\barq(p), p),~~\mbox{for a.e. $p\in(0, 1)$.}
\end{align*}
Both sides of the above equation are right-continuous, so
\begin{align} \label{H'continuous}
\barh'(p)= \frac{\partial V}{\partial x}(\barq(p), p),~~\mbox{for $p\in(0, 1)$.}
\end{align}
Moreover, the jump points of $\barh'$ and $\barq$ coincide, and $\barh'$ jumps downward when $\barq$ jumps (upward), that is, $\barh'(p-)>\barh'(p)$ if $\barq(p)>\barq(p-)$ at some $p\in(0,1)$. 

In terms of $\barh$, the condition \eqref{varcond} now reads
\begin{align}\label{condition2}
\int_{0}^{1} \big(\barh(p)-\lam\eta(p)\big)'(Q(p)-\barq(p))\ddp\leq 0,~~ \forall\; Q \in \setq.
\end{align}
By taking $Q=2\barq$ and $Q=\frac{1}{2}\barq$ in above, we see the condition \eqref{varcond} is equivalent to the following two conditions
\begin{align} \label{seta-1}
\int_{0}^{1} \big(\barh(p)-\lam\eta(p)\big)'\barq(p)\ddp=0, 
\end{align}
and
\begin{align} \label{seta-2}
\int_{0}^{1} \big(\barh(p)-\lam\eta(p)\big)'Q(p)\ddp\leq 0,~~ \forall\; Q \in \setq.
\end{align}

For any $a\in [0, 1)$, taking $Q(p)= \idd{p\in [a, 1)}\in\setq$ in \eqref{seta-2} yields 
\begin{align*}
\int_{0}^{1}\big(\barh(p)-\lam\eta(p)\big)' \idd{p\in [a, 1)}\ddp=-\big(\barh(a)-\lam\eta(a)\big)\leq 0, 
\end{align*}
on recalling $\barh(1)=\eta(1)=0$.
It hence follows $\barh(p)\geq \lam\eta(p)$ on $[0, 1]$. 

\begin{lemma}
The optimal solution $\barq$ to the problem \eqref{static4}, if exists, must be a continuous increasing function. 
As a consequence, $\barh'$ is continuous.
\end{lemma}
\begin{proof}
Suppose, on the contrary, $\barq(a)-\barq(a-)>\ep>0$ at some point $a\in(0,1)$. 
Let
\begin{align*}
Q_{\pm}(p)=
\begin{cases}
\barq(p)\pm\ep, & p\in[a, 1),\\
\barq(p), & p\in(0,a).
\end{cases} 
\end{align*}
Then $Q_{\pm}\in\setq$. By \eqref{varcond} and $\barh(1)=\eta(1)=0$, we have 
\begin{multline*}
~\int_{0}^{1}\big(\barh(p)-\lam\eta(p)\big)'\big(Q_{\pm}(p)-\barq(p)\big)\ddp\\
=~\pm\ep\int_{a}^{1}\big(\barh(p)-\lam\eta(p)\big)' \ddp=~\mp\ep \big(\barh(a)-\lam\eta(a)\big)\leq 0,
\end{multline*} 
so $\barh(a)=\lam\eta(a)$. This together with $\barh\geq \lam\eta$ leads to $\barh'(a)\geq \lam\eta'(a)\geq \barh'(a-)$. But, since $\barq(a)>\barq(a-)$, we get from \eqref{H'continuous} that $\barh'(a-)>\barh'(a)$, leading to a contradiction. Therefore, $\barq$ is continuous. This together with \eqref{H'continuous} implies the continuity of $\barh'$, completing the proof. 
\end{proof}

Suppose $\barh>\lam\eta$ on some interval $(a, b]\subset (0, 1)$ with $\barq(b)>\barq(a)$. Using that quantiles are increasing, $\barh\geq \lam\eta$, $\barh(1)=\eta(1)=\barq(0)=0$, we obtain from Fubini's theorem that 
\begin{align*}
\int_{0}^{1} \big(\barh(p)-\lam\eta(p)\big)'\barq(p)\ddp =&~\int_{0}^{1}\big(\barh(p)-\lam\eta(p)\big)'\int_{(0,p]}\dd\barq(s)\ddp\\
=&~ \int_{(0, 1)}\int_s^1\big(\barh(p)-\lam\eta(p)\big)'\ddp\dd \barq(s)\\
=&~-\int_{(0, 1)} \big(\barh(s)-\lam\eta(s)\big)\dd \barq(s)\\
\leq &~ -\int_{(a, b]} \big(\barh(s)-\lam\eta(s)\big)\dd \barq(s)<0, 
\end{align*}
contradicting \eqref{seta-1}. This implies that $\barq$ is constant (and consequently $\barq'=0$) on every subinterval of the set $\big\{p\in(0, 1)\;\big|\;\barh(p)>0\big\}$.

Together with \citelem{initial}, we conclude that if $\barq$ is an optimal solution to the problem \eqref{static4}, then it is continuous increasing and satisfies 
\begin{align} \label{ode1a}
\begin{cases}
\dsp \min\Big\{\barq'(p), \barh(p)-\lam\eta(p)\Big\}=0, ~~ 
\dsp \barh'(p)=\frac{\partial V}{\partial x}(\barq(p), p), \; \mbox{for a.e. $p\in(0, 1)$;}\medskip\\
\barh(1)=0,~~\barq(p)=0,~~p\in(0,\barp].
\end{cases}
\end{align}
This is a two-dimensional first-order ODE and solvable numerically.
Based on it, we will present numerical experiments in the following Section \ref{sec:numerical}.

We now present the main theoretical result of this paper.
\begin{thm}[Characterization of Optimal Solution]\label{main1}
A function $\barq$ is the optimal solution to the problem \eqref{static4} if and only if it is a continuous increasing function on $(0, 1)$ and satisfies the ODE system \eqref{ode1a}.
\end{thm}
As a consequence, 
since the problem \eqref{static4} admits at most one optimal solution, the ODE system \eqref{ode1a} admits at most one continuous increasing solution. 
\begin{proof} 
($\Longrightarrow$) This was proved by the preceding analysis. 

($\Longleftarrow$) We now suppose $\barq$ is a solution to the ODE system \eqref{ode1a}. Then \eqref{ode1a} implies 
\begin{align*}%\label{ineq4}
\barq'(p)\ge0, ~~ \barh(p)\ge \lam\eta(p), ~~\big(\barh(p)-\lam\eta(p)\big)\dd \barq(p)=0,~\mbox{ for a.e. $p\in(0, 1)$.}
\end{align*}
Since $\barh$ and $\eta$ are continuous functions, the above implies 
\begin{align}\label{ineq4} 
\barh(p)\ge \lam\eta(p)~ \mbox{ for all $p\in[0, 1]$.}
\end{align}
In particular, it follows from $\barh(1)=\eta(1)=\barq(0)=0$ and Fubini's theorem that 
\begin{align} 
\int_{0}^{1} \big(\barh(p)-\lam\eta(p)\big)'\barq(p)\ddp=-\int_{0}^{1} \big(\barh(p)-\lam\eta(p)\big)\dd\barq(p)=0, 
\end{align}
showing that \eqref{seta-1} holds. We next use a monotone argument to prove that \eqref{seta-2} holds as well. 

First, for any $a\in [0, 1)$, we have 
\begin{align*}
\int_{0}^{1}\big(\barh(p)-\lam\eta(p)\big)' \idd{p\in [a, 1)}\ddp=-\big(\barh(a)-\lam\eta(a)\big)\leq 0, 
\end{align*}
on recalling $\barh(1)=\eta(1)=0$ and \eqref{ineq4}. Therefore, the inequality in \eqref{seta-2} holds for all $Q(p)= \idd{p\in [a, 1)}\in \setq$ with $a\in [0, 1)$. 
By linear combination, we see the inequality in \eqref{seta-2} also holds for all non-decreasing step functions $Q \in \setq$.
Finally, for any $Q \in \setq$ and $n>1$, we define a sequence of functions as: 
\begin{align*}
Q_n(p)=\begin{cases}
\frac{j}{2^n} &~\mbox{if $Q(p)\in[\frac{j}{2^n},\frac{j+1}{2^n}) $ for some $0\leq j<n2^n$;}\\
n &~\mbox{if $Q(p)\geq n$.}
\end{cases}
\end{align*}
For each $n>1$, $Q_n$ is a non-decreasing step function in $\setq$, so the inequality in \eqref{seta-2} holds for it, i.e.,
\begin{align*} 
\int_{0}^{1}\big(\barh(p)-\lam\eta(p)\big)' Q_n(p)\ddp &\leq 0.
\end{align*}
Since \[0\leq Q(p)\wedge n-Q_n(p)\leq \frac{1}{2^n},~~p\in(0, 1),\]
it follows 
\begin{align*} 
&\quad\;\int_{0}^{1}\big(\barh(p)-\lam\eta(p)\big)' \big(Q(p)\wedge n\big)\ddp \\
&= \int_{0}^{1}\big(\barh(p)-\lam\eta(p)\big)' Q_n(p)\ddp+\int_{0}^{1}\big(\barh(p)-\lam\eta(p)\big)'\big(Q(p)\wedge n-Q_n(p)\big)\ddp\\
&\leq \int_{0}^{1}\big(\barh(p)-\lam\eta(p)\big)'\big(Q(p)\wedge n-Q_n(p)\big)\ddp\\
&\leq \frac{1}{2^n}\int_{0}^{1} \Big|\big(\barh(p)-\lam\eta(p)\big)'\Big|\ddp,
\end{align*}
i.e., 
\begin{multline*} \qquad
\int_{0}^{1} \frac{\partial V}{\partial x}(\barq(p), p) \big(Q(p)\wedge n\big)\ddp-\lam\int_{0}^{1} Q_\rho(1-p) \big(Q(p)\wedge n\big)\ddp\\
\leq \frac{1}{2^n}\int_{0}^{1} \Big|\big(\barh(p)-\lam\eta(p)\big)'\Big|\ddp.\qquad
\end{multline*}
This estimate together with $\lam Q_\rho>0$ implies 
\begin{multline}\label{ineq3}\qquad
\int_{0}^{1} \frac{\partial V}{\partial x}(\barq(p), p) \big(Q(p)\wedge n\big)\ddp-\lam\int_{0}^{1} Q_\rho(1-p) Q(p)\ddp\\
\leq \frac{1}{2^n} \int_{0}^{1} \Big|\big(\barh(p)-\lam\eta(p)\big)'\Big|\ddp.\qquad
\end{multline}
Thanks to $Q_\rho>0$ and the estimate \eqref{ineq1},
\begin{multline*} 
\int_{0}^{1} \Big|\big(\barh(p)-\lam\eta(p)\big)'\Big|\ddp
\leq ~\int_{0}^{1}\bigg|\frac{\partial V}{\partial x}(\barq(p), p)\bigg|\ddp+\int_{0}^{1}\lam Q_\rho(1-p)\ddp\\
\leq ~u'\left(\qpr(0)\right)+\lam\BE{\rho}<\infty.
\end{multline*}
This, via sending $n\to\infty$ in \eqref{ineq3}, leads to 
\begin{align*} 
\int_{0}^{1} \frac{\partial V}{\partial x}(\barq(p), p) Q(p) \ddp-\lam\int_{0}^{1} Q_\rho(1-p) Q(p)\ddp\leq 0,
\end{align*}
that is, the inequality in \eqref{seta-2} holds for any $Q\in\setq$. 

Since both \eqref{seta-1} and \eqref{seta-2} hold, we have \eqref{varcond} holds, which implies the optimality of $\barq$ by \citeprop{variation}.
This completes the proof. 
\end{proof}

Recall that $\vd$ is defined as in \citelem{lemma:G_decreasing_via_IFT}.
By \eqref{ode1a},
\begin{align*}
\barq(p)=\vd(\barh'(p), p).
\end{align*}
Taking the derivative with respect to $p$ on both sides yields 
$$
\barq'(p) = \barh''(p)\frac{\partial \vd}{\partial x}(\barh'(p), p) + \frac{\partial \vd}{\partial p}(\barh'(p), p),~\mbox{ for a.e. $p\in(0, 1)$.}
$$
Then we obtain from \eqref{ode1a} that 
$$\min\left\{ \barh''(p) \frac{\partial \vd}{\partial x}(\barh'(p), p) + \frac{\partial \vd}{\partial p}(\barh'(p), p), ~\barh(p)-\lam\eta(p)\right\} = 0, ~\mbox{ for a.e. $p\in(0, 1)$.}$$
Since $\frac{\partial \vd}{\partial x}<0$, by \cite[Lemma 4.4]{X25}, the above can be rewritten as
$$\min\Big\{{}-\barh''(p)+\opl(\barh'(p), p), \; \barh(p)-\lam\eta(p)\Big\}=0, ~\mbox{ for a.e. $p\in(0, 1)$.}$$
where
\begin{align} \label{def:opl}
\opl(x,p)=-\frac{\frac{\partial \vd}{\partial p}(x,p)}{\frac{\partial \vd}{\partial x}(x,p)}
&=\frac{\partial^2V}{\partial x\partial p}\bigl(\vd(x,p),p\bigr).
\end{align} 
Note $\barq(p)= 0$ for $p\in(0,\barp]$ and $\barh(1) = 0$, so 
the two-dimensional first-order ODE system \eqref{ode1a} reduces to a one-dimensional second-order ODE: 
\begin{align} \label{ode3}
\begin{cases}
\dsp \min\Big\{{}-\barh''(p)+\opl(\barh'(p), p), \; \barh(p)-\lam\eta(p)\Big\}=0, \; \mbox{for a.e. $p\in(\barp, 1)$;}\medskip\\
\dsp\barh(1)=0,~~\barh'(p)=\frac{\partial V}{\partial x}(0, p), \; \mbox{for $p\in(0,\barp]$.} 
\end{cases}
\end{align}
This is a variational inequality with mixed boundary conditions, which can be numerically solved. 

Summarizing all the preceding results, we derive the main result of this paper as follows. 
\begin{thm}[Optimal Solutions to \eqref{static3a} and \eqref{static2}]\label{mainthm}
Suppose $\barh$ is a solution of \eqref{ode3}. Define $\barq(p)=\vd(\barh'(p), p)$ for $p\in(0,1)$, where $\vd$ is given in \citelem{lemma:G_decreasing_via_IFT}.
If 
\[\int_0^1 \barq(p) Q_\rho(1-p)\ddp = x,\]
then $\barq$ is the unique optimal solution to \eqref{static3a}.
Moreover, $\barx=\barq(1-F_{\rho}(\rho))$ is the unique optimal solution to the problem \eqref{static2}.
\end{thm}

\section{Numerical Study}\label{sec:numerical}
This section presents numerical experiments that illustrate our theoretical results and provide economic insights. Our findings demonstrate that the presence of an intractable claim significantly affects the optimal payoff structure. 

\paragraph{Numerical Methodology.}
We solve the two-dimensional first-order ODE system \eqref{ode1a} for $(\barh, \barq)$ using a forward Euler scheme with penalization to enforce constraints. The solution procedure starts with initial conditions $\barq(0)=0$ and $\barq'(0)=e$, where the parameter $e$ is tuned to satisfy the terminal boundary condition $H(1)=0$. 

\paragraph{Model Specification.}
In our numerical studies, we adopt the following specifications: 

\begin{compactitem}
\item {Intractable Claim $\pr$}: Uniform distribution on $[0,y]$ with quantile
\[
Q_{\pr}(p) = y p, \quad p\in(0,1).
\]

\item {Utility Function}: Mixture of two exponential utilities
\[
u(x) = -c_1 e^{-\gamma_1 x} - c_2 e^{-\gamma_2 x}.
\]

\item {Pricing Kernel $\rho$}: Lognormal distribution with $\log\rho \sim N(\mu_{\log\rho},\sigma^2_{\log\rho})$, where
\[
\mu_{\log\rho} = -\Big(r+\frac12\theta^2\Big)T, \quad
\sigma_{\log\rho} = |\theta|\sqrt{T},
\]
with $\theta$ the market price of risk, $r$ the interest rate, and $T$ the maturity. Its quantile is
\[
Q_\rho(p) = \exp\big(\mu_{\log\rho} + \sigma_{\log\rho}\Phi^{-1}(p)\big), \quad p\in(0,1),
\]
where $\Phi^{-1}$ is the standard normal quantile. 
\end{compactitem}

Unless otherwise stated, we fix $r=0.02$, $T=1$, and $y=2$.

\subsection{Quantile of Pricing Kernel under Different $\theta$}

Figure~\ref{fig1_theta} displays the quantile function $Q_\rho$ of the pricing kernel $\rho$ under different market prices of risk $\theta$, with investor parameters fixed as
\[
\alpha = 0.25,\quad c_1 = 950,\quad c_2 = 950,\quad \gamma_1 = 0.010,\quad \gamma_2 = 0.012.
\]

{The quantile increases more rapidly as $\theta$ grows, and approaches a flatter shape for smaller $\theta$. This reflects the increased dispersion of state prices in more volatile market environments.}

\begin{figure}[H]
\centering
\includegraphics[width=0.68\textwidth]{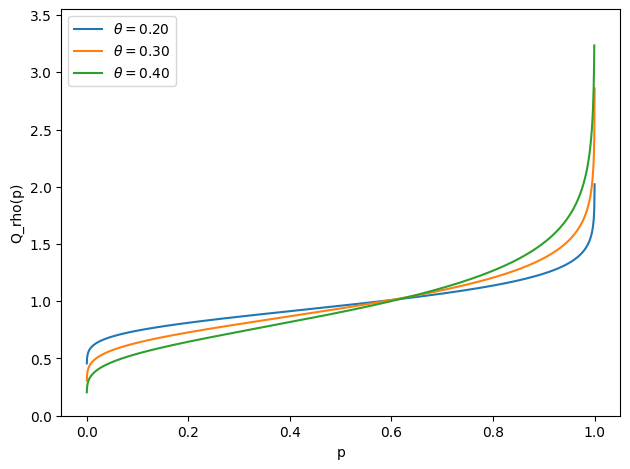}
\caption{Quantile function $Q_\rho(p)$ of the pricing kernel $\rho$ under different market prices of risk $\theta$.}
\label{fig1_theta}
\end{figure}

\subsection{Relationship between Initial Endowment $x$ and Lagrange Multiplier $\lambda$}

Recall that problem \eqref{static3a} imposes the budget constraint
\[
\int_0^1 Q(p) Q_\rho(1-p)\ddp = x.
\]
We remove this constraint via the Lagrange multiplier method, leading to problem \eqref{static4}. Figure~\ref{fig2} illustrates the one-to-one correspondence between the initial endowment $x$ and the Lagrange multiplier $\lambda$.

{As expected from concavity of the quantile optimization problem, the relationship is monotonically decreasing, indicating a well-behaved mapping. The nonlinear shape reflects the underlying complexity of the problem.}

\begin{figure}[H]
\centering
\includegraphics[width=0.68\textwidth]{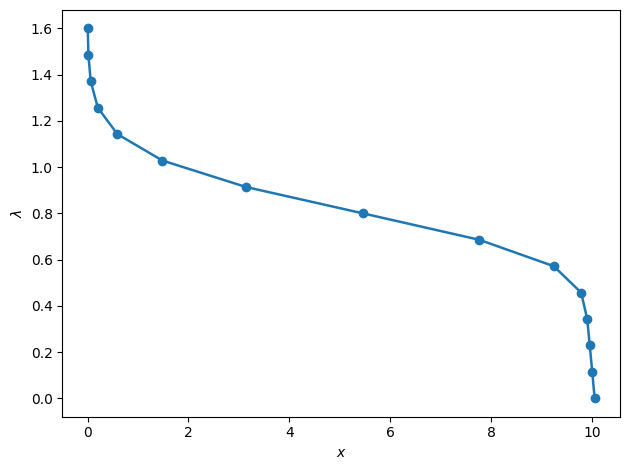}
\caption{Relationship between initial endowment $x$ (horizontal axis) and Lagrange multiplier $\lambda$ (vertical axis).}
\label{fig2}
\end{figure}

\subsection{Optimal Payoff under Different $\theta$}

We now examine how market conditions affect the optimal payoff. Fix investor parameters as
\[
x = 7.66,\quad \alpha = 0.25,\quad c_1 = 950,\quad c_2 = 950,\quad \gamma_1 = 0.010,\quad \gamma_2 = 0.012.
\]

Figure~\ref{fig4_theta} plots the optimal payoff profiles $\rho \mapsto \barq(1-F_\rho(\rho))$ under different $\theta$. The market price of risk $\theta$ captures the compensation per unit risk; a higher $\theta$ implies greater dispersion in state prices across market scenarios.

{Holding the budget fixed, increasing $\theta$ alters the relative cost of transferring wealth across states, thereby reshaping the optimal payoff profile. A larger $\theta$ improves performance in both the good (small $\rho$) and bad (large $\rho$) market states, while reducing it in intermediate states. These differences are most pronounced in the tails, where state prices are extreme --- consistent with dispersion amplifying marginal pricing differences between favorable and unfavorable states.}

\begin{figure}[H]
\centering
\includegraphics[width=0.68\textwidth]{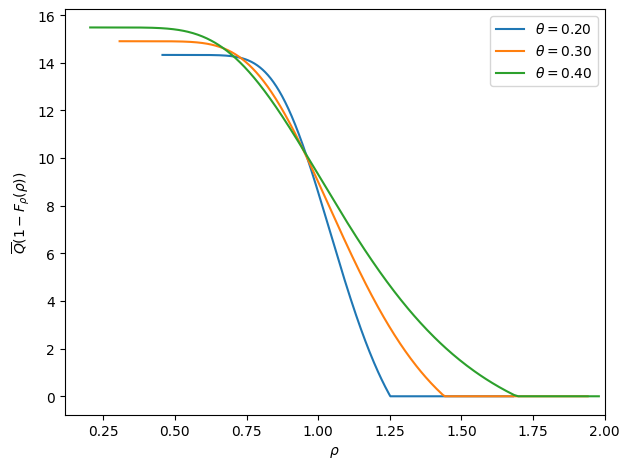}
\caption{Optimal payoff profiles $\rho \mapsto \barq(1-F_\rho(\rho))$ under different market prices of risk $\theta$.}
\label{fig4_theta}
\end{figure}

\subsection{Sensitivity Analysis with Respect to Investor Parameters}

{Having examined market effects, we now fix the market environment and study how investor characteristics influence optimal payoffs.}

\subsubsection{Impact of Initial Endowment $x$}

Fix investor parameters as
\[
\alpha = 0.25,\quad c_1 = 950,\quad c_2 = 950,\quad \gamma_1 = 0.010,\quad \gamma_2 = 0.012.
\]

Figure~\ref{fig3} displays optimal payoff profiles under different initial endowments $x$. {Larger budgets shift the entire profile upward, particularly in favorable states where transferring wealth is cheaper. This reflects the relaxation of the budget constraint, allowing greater payoff delivery in states with lower pricing kernel values.}

\begin{figure}[H]
\centering
\includegraphics[width=0.68\textwidth]{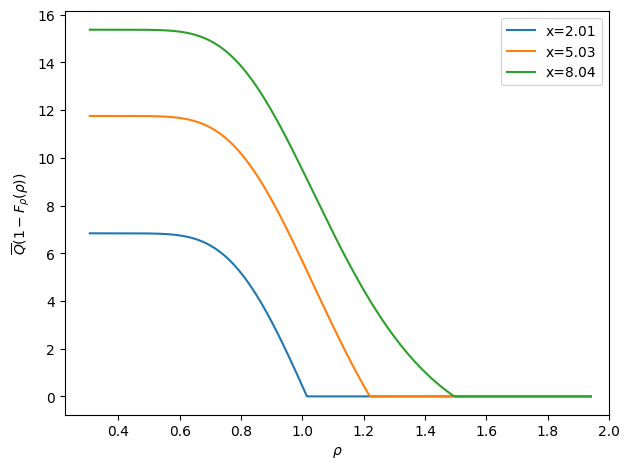}
\caption{Optimal payoff profiles $\rho \mapsto \barq(1-F_\rho(\rho))$ under different initial endowments $x$.}
\label{fig3}
\end{figure}

\subsubsection{Impact of Utility Function Parameters}

Fix $y = 8$, and investor parameters as
\[
x=6.26,\quad\alpha = 0.25,\quad c_1 = 950,\quad c_2 = 950.
\]
Figure~\ref{fig7} illustrates the effect of utility parameters on optimal payoff profiles. {Smaller risk aversion parameters (corresponding to smaller $\gamma_1$, $\gamma_2$ values) yield more stable performance across states. While higher risk aversion may improve performance in good market states, it comes at the cost of worse outcomes in bad states --- a classical risk-return trade-off.}

\begin{figure}[H]
\centering
\includegraphics[width=0.68\textwidth]{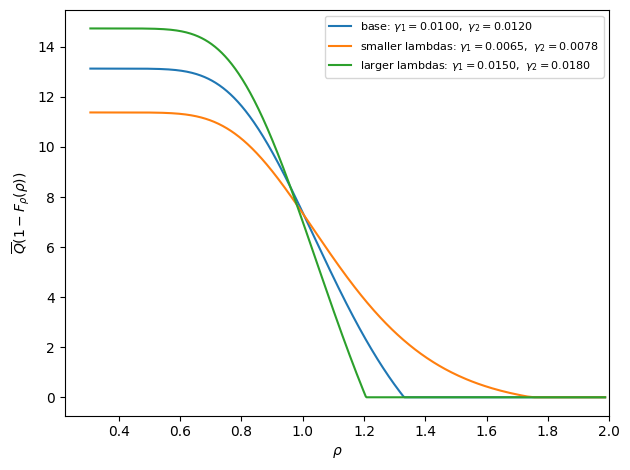}
\caption{Optimal payoff profiles $\rho \mapsto \barq(1-F_\rho(\rho))$ under different utility function parameters $\gamma_1$, $\gamma_2$. 
}
\label{fig7}
\end{figure}

Fix $y = 10$, and investor parameters as
\[
x=9.72,\quad \alpha = 0.25,\quad \gamma_1 = 0.01,\quad \gamma_2 = 0.018.
\]
Figure~\ref{fig7_2} illustrates the effect of utility parameters on optimal payoff profiles. 
\begin{figure}[H]
\centering
\includegraphics[width=0.68\textwidth]{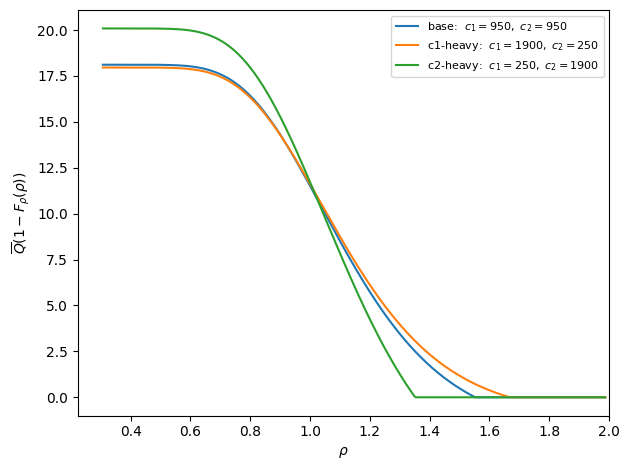}
\caption{Optimal payoff profiles $\rho \mapsto \barq(1-F_\rho(\rho))$ under different utility function parameters $c_1$, $c_2$. 
}
\label{fig7_2}
\end{figure}
\subsubsection{Impact of Ambiguity Attitude $\alpha$}

To examine the effect of the ambiguity attitude $\alpha$, we set $y=20$ and fix investor parameters as
\[
x = 15.17,\quad c_1 = 200,\quad c_2 = 3600,\quad \gamma_1 = 0.0008,\quad \gamma_2 = 0.0800.
\]
Figure~\ref{fig5} shows that higher $\alpha$ (greater optimism) improves performance, particularly in good market scenarios, although the differences are modest under this parameter specification. 

\begin{figure}[H]
\centering
\includegraphics[width=0.68\textwidth]{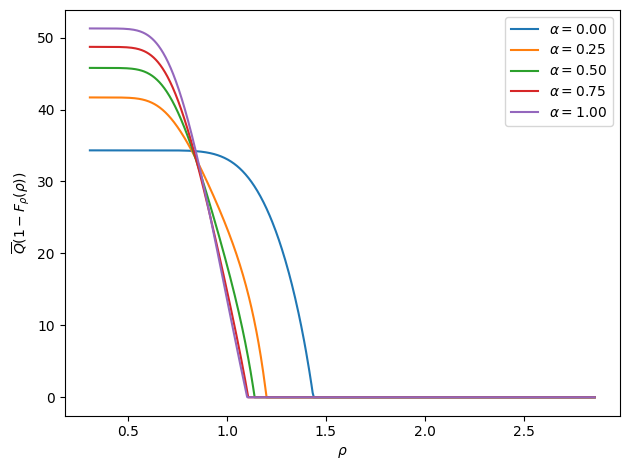}
\caption{Optimal payoff payoffs $\rho \mapsto \barq(1-F_\rho(\rho))$ under different ambiguity attitudes $\alpha$.}
\label{fig5}
\end{figure}

\subsection{Impact of Intractable Claim Distribution}

Finally, we investigate how the distribution of the intractable claim $\pr$ affects optimal payoffs. We fix $$x=0.17,~\alpha=0.6,~c_1= 0.5,~c_2 = 0.5,~\gamma_1 = 1, \gamma_2= 2,$$ 
and assume $\pr$ follows a normal distribution $N(\mu_\pr,\sigma_\pr)$ truncated on $[a,b]$, written as TN$(\mu_\pr,\sigma_\pr)[a,b]$ in Figure~\ref{fig6}.

The results demonstrate that both location and dispersion parameters significantly influence the optimal payoff profile. Increasing the mean $\mu_\pr$ shifts the curve upward, while increasing the standard deviation $\sigma_\pr$ amplifies curvature and heterogeneity across quantiles. This is consistent with greater dispersion in the claim distribution feeding into the marginal value of wealth $\frac{\partial V}{\partial x}$ and thereby reshaping the optimal allocation across states. 

\begin{figure}[H]
\centering
\includegraphics[width=0.68\textwidth]{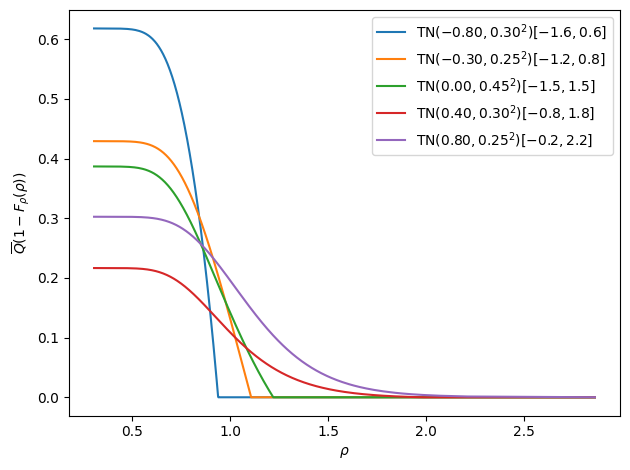}
\caption{Optimal payoff profiles $\rho \mapsto \barq(1-F_\rho(\rho))$ under different intractable claim distributions.}
\label{fig6}
\end{figure}

In summary, our numerical experiments confirm that the intractable claim significantly affects optimal investment strategies. The quantile-based approach successfully captures these effects and reveals how market conditions, investor preferences, and claim characteristics interact in determining optimal payoffs. 

\section{Concluding remarks}\label{sec:conclusion}

This paper studied an $\alpha$-robust utility maximization problem in the presence of an intractable claim --- an exogenous contingent claim with known marginal distribution but unspecified dependence structure with the financial market. By combining rearrangement theory with quantile optimization methods, we transformed the original dynamic stochastic control problem into a concave static optimization over quantile functions. The optimal quantile was characterized via a two-dimensional first-order ODE system, which we solved numerically to obtain economic insights. 

Our framework makes several contributions. First, it generalizes both worst-case \cite{LXZ23} and best-case evaluations through a continuous parameter $\alpha$ that captures the investor's ambiguity attitude. Second, it demonstrates that law-invariance of the $\alpha$-robust risk measure, established via comonotonicity theory, permits a tractable quantile reformulation without requiring any assumptions on the joint dependence structure. Third, our numerical studies reveal how market conditions, investor preferences, and the distribution of the intractable claim interact to shape optimal payoffs. 

The quantile-based approach developed here naturally accommodates additional risk constraints. Popular risk measures such as Value-at-Risk (VaR) and Expected Shortfall (ES) can be incorporated as explicit constraints on the quantile function. While we focused on the unconstrained problem to elucidate the core economic trade-offs, the extension to constrained settings represents a promising direction for future research. Many interesting extensions would further enhance the applicability of robust optimization methods in portfolio choice under dependence uncertainty.

%\newpage

\end{document}